\documentclass[9pt,technote,letterpaper]{IEEEtran}
\usepackage{cite}
\usepackage{amsmath,amssymb,amsfonts,amsthm}
\usepackage{algorithmic}
\usepackage{graphicx}
\usepackage{textcomp}
\usepackage{subfigure}
\usepackage{algorithm}
\newtheorem{definition}{Definition}
\newtheorem{assumption}{Assumption}
\newtheorem{theorem}{Theorem}
\newtheorem{lemma}{Lemma}
\def\BibTeX{{\rm B\kern-.05em{\sc i\kern-.025em b}\kern-.08em
		T\kern-.1667em\lower.7ex\hbox{E}\kern-.125emX}}
\begin{document}
	\title{Decentralized Multi-Agent Reinforcement Learning: An Off-Policy Method}
	\author{Kuo Li, \IEEEmembership{Student Member, IEEE}, Qing-Shan Jia, \IEEEmembership{Senior Member, IEEE}
		\thanks{K. Li and Q.-S. Jia are with the Center for Intelligent and Networked System, Department of Automation, Beijing National Research Center for Information Science and Technology (BNRist), Tsinghua University, Beijing 100084, P.R. China (e-mail: li-k19@mails.tsinghua.edu.cn, jiaqs@tsinghua.edu.cn)}}
	
	\maketitle
	
	\begin{abstract}
		We discuss the problem of decentralized multi-agent reinforcement learning (MARL) in this work. In our setting, the global state, action, and reward are assumed to be fully observable, while the local policy is protected as privacy by each agent, and thus cannot be shared with others. There is a communication graph, among which the agents can exchange information with their neighbors. The agents make individual decisions and cooperate to reach a higher accumulated reward.
		
		Towards this end, we first propose a decentralized actor-critic (AC) setting. Then, the policy evaluation and policy improvement algorithms are designed for discrete and continuous state-action-space Markov Decision Process (MDP) respectively. Furthermore, convergence analysis is given under the discrete-space case, which guarantees that the policy will be reinforced by alternating between the processes of policy evaluation and policy improvement. In order to validate the effectiveness of algorithms, we design experiments and compare them with previous algorithms, e.g., Q-learning \cite{watkins1992q} and MADDPG \cite{lowe2017multi}. The results show that our algorithms perform better from the aspects of both learning speed and final performance. Moreover, the algorithms can be executed in an off-policy manner, which greatly improves the data efficiency compared with on-policy algorithms.
	\end{abstract}
	
	\begin{IEEEkeywords}
		consensus, multi-agent system, reinforcement learning
	\end{IEEEkeywords}
	
	\section{Introduction}
	In the past decades, reinforcement learning (RL) \cite{sutton2018reinforcement} achieves greatly attention for its excellent performance on various complex tasks, e.g., playing Go \cite{silver2016mastering}, video games \cite{mnih2013playing} and autonomous driving \cite{jin2020game}. Related study mainly focuses on single-agent problems, which can be formulated as Markov Decision Process (MDP). These algorithms can be mainly categorized into three classes, namely, value-based algorithms \cite{mnih2013playing,van2016deep}, policy-based algorithms \cite{schulman2015trust,schulman2017proximal} and actor-critic algorithms (AC) \cite{silver2014deterministic, fujimoto2018addressing, haarnoja2018soft}. AC algorithms are getting more and more popular for their outstanding performance on complex tasks, e.g., control problems of robotics \cite{lillicrap2015continuous,kober2013reinforcement}, and achieve state of the art performance.
	
	Particularly, AC algorithms are mainly composed of two parts, namely the actor and the critic. The actor encodes the policy that specifies how to act under each state. The critic evaluates the policy. Typically, both the actor and critic are approximated by deep neural networks (DNNs) and optimized by stochastic gradient descent (SGD).
	
	With the development of RL technology, the applications of AC-based algorithms gradually extend to more complex tasks, which involve cooperation or competition \cite{gupta2017cooperative,zhang2017dynamic,chen2021multi} among multiple agents. Some technical problems may arise. Firstly, as the number of agents increases, the decision space expands exponentially, which causes the curse of dimensionality. Then, more importantly, in some real applications of multi-agent systems, policies will be protected locally as privacy, thus cannot be optimized globally. Therefore, traditional single-agent RL algorithms may not work well in the multi-agent environment, thus many MARL algorithms are further proposed to solve these problems \cite{zhang2018fully,lowe2017multi}. In this work, we focus on cooperative MARL problems with the AC algorithms.
	
	In cooperative MARL problems, most related algorithms are designed as centralized training and decentralized execution (CTDE), e.g., QMIX \cite{rashid2018qmix} and MADDPG \cite{lowe2017multi}. CTDE enables the application of off-policy algorithms by sharing information among agents. Off-policy algorithms evaluate current policy by trajectories sampled from any distribution, which means historical experience can be reused for current evaluation and higher data efficiency compared with on-policy algorithms, which requires the trajectories used for evaluation must be sampled from current policy. However, the centralized training process also suffers the curse of dimensionality, since the optimization process requires global action, whose size increases with the number of agents. To ease the burden, MARL algorithms executed in a distributed manner are proposed by researchers. For example, Tan \cite{tan1993multi} develops independent Q-learning (IQL), which independently applies a Q-learning algorithm to each agent. Despite the good performance on some tasks, its convergence cannot be guaranteed, since the interplay between agents is not considered. Jia et al \cite{jia2020event} propose to use the event-based optimization (EBO) algorithm to solve the policy optimization problem in a distributed manner. In their work, a networked multi-agent system with imperfect communication channel is considered. They construct an EBO model and design algorithms based on the performance difference equation and the performance derivative equation. The EBO model decreases the size of the decision space and benefits policy searching. Besides, Zhang et al \cite{zhang2018fully} also propose decentralized algorithms with networked agents. During either the training or the execution process, each agent only needs to communicate with its neighbors, rather than collecting global messages. In their work, the training process is decomposed into three steps, namely, critic step, which updates each critic with only local information; consensus step, which adjusts each critic to evaluate global policy by exchanging information among neighbors; and actor step, which improves the local policy with the local critic. The training process of that algorithm is executed in an on-policy manner since global policy cannot be used locally for each agent. Compared with the off-policy manner, it may lose some data efficiency.
	
	In order to improve the data efficiency and provide the convergence guarantee, we propose decentralized MARL algorithms, which can be executed in a distributed and off-policy manner. In our algorithms, each agent optimizes its own AC models to push the global policy to reach better performance. Similar to work in \cite{zhang2018fully}, our algorithms also run with three steps, namely, critic step, consensus step, and actor step. However, different from \cite{zhang2018fully}, in the critic step, each agent only evaluates its own policy by treating others as part of the environment. Therefore, the training process can be executed in an off-policy manner, and the historical experience can be reused as in DQN \cite{mnih2013playing}. Since other agents are treated as part of the environment from the perspective of a given agent, and the policies of all the agents change gradually during the training process, directly reusing historical experience to evaluate current policy will make the training process unstable. In the consensus step, we use the importance sampling (IS) \cite{glynn1989importance} technique to solve this problem. In the actor step, each agent optimizes local policy with its own critic. Firstly, for discrete state-action-space MDP (referring to the MDP has finite state-action pairs in this paper), we use Q-table and policy-table to encode the critic and actor respectively. Corresponding algorithm and convergence analysis are given. Then we extend it to the continuous state-action-space case by replacing the tables with neural networks (NNs). The performance of the proposed algorithms for both cases is evaluated by experiments and it is shown that our algorithms perform competitively from the aspects of both convergence rate and final score. The reinforcement learning algorithm proposed for the discrete state-action-space MDP is extended from IQL \cite{tan1993multi}. We prove that our algorithm can converge to satisfied solutions by setting appropriate learning rates for both the policy evaluation and policy improvement processes.
	
	The contributions of this paper are mainly as follows,
	\begin{enumerate}
		\item A fully decentralized MARL model is proposed, with which both the policy evaluation and policy improvement process can be executed in a distributed manner, and corresponding algorithms for both discrete and continuous state-action-space MDP are proposed;
		\item Under the discrete state-action-space case, convergence analysis is established, which guarantees that the policy will be improved by alternately applying the algorithms of policy evaluation and policy improvement. Under the continuous state-action-space case, the algorithm is designed to be executed in an off-policy manner, which has higher data efficiency.
		\item We design experiments to show their effectiveness, and by comparing with related algorithms, our algorithms are shown to perform better from the aspects of both learning speed and final performance.
	\end{enumerate}
	
	The rest of this paper is organized as follows. We firstly introduce the basics of MDP and the multi-agent setting considered in our work in Section \ref{sec:pre}. Then in Section \ref{sec:fsMDP}, we propose the distributed MARL algorithms for discrete state-action-space MDP followed by its convergence analysis. After that, the algorithm is extended to continuous state-action-space case in Section \ref{sec:csMDP}. Finally, we test our algorithms with experiments in Section \ref{sec:sr} and reach the conclusion in Section \ref{sec:con}.
	
	\section{Preliminary}
	\label{sec:pre}
	We begin by introducing the basics of MDP and the multi-agent setting considered in this work. After that, a decentralized AC model is proposed.
	
	\subsection{MDP and multi-agent setting}
	An MDP can be characterized as a tuple $\mathcal{M}=\langle \mathcal{S},\mathcal{A},P,R,\gamma \rangle$, where $\mathcal{S}$ is the state space and $\mathcal{A}$ is the action space. $P$ is the state transition matrix. And $P(s,a,\tilde{s}):\mathcal{S}\times\mathcal{A}\times\mathcal{S}\to[0,1]$ defines the state transition probability from state $s$ to $\tilde{s}$ by taking action $a$. $R(s, a):\mathcal{S}\times\mathcal{A}\to\mathbb{R}$ defines the reward function when taking action $a$ under state $s$. Typically, $R$ is assumed to be bounded. Finally, $\gamma\in[0, 1)$ is the discounted factor for balancing current and future rewards.
	
	The following multi-agent setting is considered in our work.
	\begin{definition}
		There are totally $N$ agents connected by a communication network $\mathcal{G}=\left\{\mathcal{X},\mathcal{E}\right\}$, where $\mathcal{X}$ is the vertex set and $\mathcal{E}$ is the edge set. Besides, $\mathcal{G}$ is a connected graph. In this multi-agent system, the global action $a_t$ can be decomposed as $a_t=(a_t^1, a_t^2, \cdots, a_t^N)\in \mathcal{A}$, where $a_t^i$ is the action of the agent $i$ at time $t$. As \cite{zhang2018fully}, it is assumed that the system state $s_t\in \mathcal{S}$ is fully observable for each agent, and the instantaneous reward depends on the current state $s_t$ as well as the joint action $a_t$, i.e., $r_t=R(s_t,a_t)$, and is shared by all of the agents.
	\end{definition}
	
	Besides, each agent only has permission to use its own local policy. For the cause of privacy, the policy of other agents can be used in neither training nor execution processes.
	
	\subsection{Actor-Critic}
	\label{sec:ac}
	Here we briefly introduce AC, which is popular in RL. In AC, the actor defines the policy and the critic estimates the discounted future reward.
	
	Probabilistic policy $\pi(s,a):\mathcal{S}\times\mathcal{A}\to \mathbb{R}$ is considered in this work, which defines the probability (for the case of discrete action space) or probability density (for the case of continuous action space) of taking action $a$ under state $s$. Each agent maintains its local policy $\pi^i(s,a^i):\mathcal{S}\times\mathcal{A}^i\to \mathbb{R}$, and the global policy can be formed correspondingly $\pi(s,a)=\prod{i=1}^N \pi^i(s, a^i)$. Besides, it is assumed that following condition is satisfied.
	\begin{assumption}
		\label{ass:gt0}
		For each state-action pair, $\pi(s,a)>0$, which is equivalent to
		\begin{equation}
		\pi^i(s,a^i)>0,\forall i, s, a^i.
		\end{equation}
	\end{assumption}
	
	It guarantees that each state-action pair will be visited infinitely often, which is essential for the convergence of the critics.
	
	We assign a critic (action-value function) to each agent $Q_{\pi}^{i}(s,a^i): \mathcal{S}\times\mathcal{A}^i\to \mathbb{R}$, which tracks the following term
	\begin{equation}
	\label{equ:q_i}
	Q_{\pi}^i(s,a^i)=\mathop{\mathbb{E}}\limits_{\pi}\left[\sum_{t=0}^\infty\gamma^t {R}(s_t,a_t)|s_0=s,a_0^i=a^i\right],
	\end{equation}
	which the expected accumulated reward starting from $(s, a^i)$ at agent $i$ under current global policy $\pi$.
	
	We further define the state-value function which tracks the expected accumulated reward starting from state $s$ under $\pi$
	\begin{equation}
	V^\pi(s)=\mathop{\mathbb{E}}\limits_{\pi}\left[\sum_{t=0}^\infty\gamma^t {R}(s_t,a_t)|s_0=s\right].
	\end{equation}
	It is easy to check that
	\begin{equation}
	\label{equ:vp}
	V^\pi(s)=\mathop{\mathbb{E}}\limits_{a^i\sim\pi^i}Q^i_{\pi}(s,a^i).
	\end{equation}
	
	Furthermore, $Q^i_{\pi}$ follows the Bellman equation
	\begin{equation}
	\label{equ:be}
	\begin{aligned}
	Q^{i}_{\pi}(s,a^i)
	=&\mathop{\mathbb{E}}\limits_{a^{-i}\sim \pi^{-i}}\left[ {R}(s,a)+\gamma\mathop{\mathbb{E}}\limits_{\tilde{s}\sim P}V^\pi(\tilde{s}) \right]\\
	=&\mathop{\mathbb{E}}\limits_{a^{-i}\sim \pi^{-i}}\left[ {R}(s,a)+\gamma \mathop{\mathbb{E}}\limits_{\substack{\tilde{s}\sim P\\\tilde{a}^i\sim\pi^i}}Q^i_{\pi}(\tilde{s},\tilde{a}^i)\right],
	\end{aligned}
	\end{equation}
	where we use superscript $-i$ to indicate the corresponding item of all the agents except agent $i$. For example, $a^{-i}$ and $\pi^{-i}$ in \eqref{equ:be} are the joint action and policy of all the agents except agent $i$.
	
	In the following sections, we will propose distributed MARL algorithms based on the multi-agent MDP and AC setting above.
	
	\section{MARL with discrete state-action-space MDP}
	\label{sec:fsMDP}
	In the scenario of discrete state-action-space MDP (referring to the MDP has finite state-action pairs in this paper), each agent $i$ maintains a Q-table, $Q^i\in \mathbb{R}^{|\mathcal{S}|\times|\mathcal{A}^i|}$, and a policy-table, $\zeta^i\in\mathbb{R}^{|\mathcal{S}|\times|\mathcal{A}^i|}$. The policy of agent $i$ is calculated by the softmax function
	\begin{equation}
	\label{equ:zeta}
	\pi^i(s,a^i)=\frac{\exp\left(\zeta^i(s,a^i)\right)}{\sum\limits_{\tilde{a}^i\in\mathcal{A}^i}\exp\left(\zeta^i(s,\tilde{a}^i)\right)}.
	\end{equation}
	
	In each step, agent $i$ takes action $a_t^i$ according to current state $s_t$ and local policy $\pi^i$. Then by applying joint action $a_t$, the system state will transit to $s_{t+1}$ and return the global reward $r_t$. The tuple $(s_t,a_t,r_t,s_{t+1})$ is used to update both $Q^i$ and $\zeta^i$, which correspond to policy evaluation step and policy improvement step respectively.
	
	\subsection{Policy Evaluation}
	\label{sec:tpe}
	
	In the policy evaluation step, we randomly initialize $Q^i$ as $Q_0^i$ and update it by
	\begin{equation}
	\label{equ:qupdate}
	\begin{aligned}
	&Q^i_{t+1}(s_t,a^i_t)\\=&Q_t^i(s_t,a_t^i)+\alpha_t^i(s_t,a_t^i)\left[r_t+\gamma Q_t^i(s_{t+1},\tilde{a}^i)-Q_t^i(s_t,a_t^i)\right],
	\end{aligned}
	\end{equation}
	where $\tilde{a}^i$ is a local action freshly sampled from $\pi^i$. $\alpha_t^i(s_t,a_t^i)\in[0,1]$ is the step size of the $(t+1)$-th update. Specially, only $Q^i(s_t,a^i_t)$ is updated and others, i.e., $Q^j(s_t,a^j_t), j\neq i$, remain the same. Notice that the policy is assumed to be fixed when analyzing the convergence of $Q^i$. The convergence can be guaranteed by Theorem \ref{the:pe}.
	
	\begin{theorem}
		\label{the:pe}
		Under Assumption \ref{ass:gt0} and given the update rule
		\begin{equation}
		\begin{aligned}
		&Q^i_{t+1}(s_t,a^i_t)\\=&Q_t^i(s_t,a_t^i)+\alpha_t^i(s_t,a_t^i)\left[r_t+\gamma Q_t^i(s_{t+1},\tilde{a}^i)-Q_t^i(s_t,a_t^i)\right],
		\end{aligned}
		\end{equation}
		$Q^i_t$ converges to $Q_\pi^i$ with probability $1$ (w.p.1), i.e., $\lim_{t\to\infty}Q_t^i=Q_\pi^i$, as long as for each state-action pair $(s,a^i)$ the step size satisfies
		\begin{equation}
		\sum\limits_{t=0}^\infty\alpha_t^i(s,a^i)=\infty,
		\end{equation}
		\begin{equation}
		\sum\limits_{t=0}^\infty\left(\alpha_t^i(s,a^i)\right)^2<\infty.
		\end{equation}
	\end{theorem}
	
	Before proving it, we introduce a lemma from stochastic approximation.
	\begin{lemma}[\cite{melo2001convergence}]
		\label{lem:sa}
		The random process $\{\Delta_t\}$ is defined as
		\begin{equation}
		\Delta_{t+1}(x)=(1-\alpha_t(x))\Delta_t(x)+\alpha_t(x)F_t(x),
		\end{equation}
		where $\alpha_t(x)$ is a scalar assigned to $x$ at time step $t$ and $F_t$ is another random process. Let $\mathcal{F}_t=\left\{F_i|i<t\right\}$, then $\Delta_t$
		converges to $0$ w.p.1, if the following assumptions are satisfied
		\begin{enumerate}
			\item\label{condition1} $0\leq\alpha_t\leq 1$,$\sum_{t=0}^{\infty}\alpha_t(x)=\infty$ and $\sum_{t=0}^{\infty}\alpha_t^2(x)<\infty$;
			\item\label{condition2} $||\mathbb{E}[F_t(x)|\mathcal{F}_t ]||_W\leq \gamma ||\Delta_t||_W$, with $\gamma<1$;
			\item\label{condition3} $var[F_t(x)|\mathcal{F}_t]\leq C(1+||\Delta_t||^2_W)$, for $C>0$.
		\end{enumerate}
	\end{lemma}
	
	Lemma \ref{lem:sa} is a common tool for analyzing the convergence of reinforcement learning algorithms, e.g., \cite{melo2001convergence,jaakkola1994convergence}, and the proof can be found in \cite{jaakkola1994convergence}.
	
	Then we start the proof of Theorem \ref{the:pe}.
	
	\begin{proof}
		The updating rule \eqref{equ:qupdate} can be rewritten as
		\begin{equation}
		\label{equ:rqu}
		\begin{aligned}
		Q_{t+1}^i(s_t,a_t^i)=&\left(1-\alpha_t^i(s_t,a_t^i)\right)Q_t^i(s_t,a_t^i)\\
		&+\alpha_t^i(s_t,a_t^i)\left[r_t+\gamma Q_t^i(s_{t+1},\tilde{a}^i)\right].
		\end{aligned}
		\end{equation}
		We construct $\Delta_t(x)$, $F_t(x)$ in Lemma \ref{lem:sa} as
		\begin{equation}
		\Delta_t^i(s,a^i)=Q_t^i(s,a^i)-Q_\pi^i(s,a^i),
		\end{equation}
		\begin{equation}
		\label{equ:Ft}
		F_t^i(s,a^i)=R(s,a)+\gamma Q_t^i(\tilde{s},\tilde{a}^i)-Q_\pi^i(s_t,a_t^i),
		\end{equation}
		respectively. In \eqref{equ:Ft}, $\tilde{s}$ is the succeeding state of $s$ after acting as $a$, and $\tilde{a}^i$ is an action sampled from $\pi^i$ under state $\tilde{s}$. Then substituting them to \eqref{equ:rqu}, yields
		\begin{equation}
		\begin{aligned}
		\Delta_{t+1}^i(s_t,a_t^i)=\left(1-\alpha_t^i(s_t,a_t^i)\right)\Delta_t^i(s_t,a_t^i)
		+\alpha_t^i(s_t,a_t^i)F_t^i(s_t,a_t^i).
		\end{aligned}
		\end{equation}
		
		After that, we prove Theorem \ref{the:pe} by verifying the conditions in Lemma \ref{lem:sa}. Condition \ref{condition1} in Lemma \ref{lem:sa} has been naturally considered and the other 2 conditions will be verified here.
		
		To verify Condition \ref{condition2}, notice that
		\begin{equation}
		\label{equ:ef}
		\begin{aligned}
		&\mathbb{E}\left[F_t^i(s,a^i)\left|\mathcal{F}_t\right.\right]\\
		=&\mathop{\mathbb{E}}\limits_{a^{-i}\sim\pi^{-i}}\left[R(s,a)+\gamma\mathop{\mathbb{E}}\limits_{\substack{\tilde{s}\sim P\\\tilde{a}^i\sim\pi^i}}Q_t^i(\tilde{s},\tilde{a}^i)-Q_\pi^i(s,a^i)\right]\\
		=&\mathop{\mathbb{E}}\limits_{a^{-i}\sim\pi^{-i}}\left[R(s,a)+\gamma\mathop{\mathbb{E}}\limits_{\substack{\tilde{s}\sim P\\\tilde{a}^i\sim\pi^i}}Q_t^i(\tilde{s},\tilde{a}^i)\right]-Q_\pi^i(s,a^i),
		\end{aligned}
		\end{equation}
		which induce an operator $H^i$ on a function $q^i:\mathcal{S}\times\mathcal{A}^i\to \mathbb{R}$
		\begin{equation}
		H^iq^i(s,a^i)=\mathop{\mathbb{E}}\limits_{a^{-i}\sim\pi^{-i}}\left[R(s,a)+\gamma\mathop{\mathbb{E}}\limits_{\substack{\tilde{s}\sim P\\\tilde{a}^i\sim\pi^i}}q^i(\tilde{s},\tilde{a}^i)\right].
		\end{equation}
		The operator $H^i$ is a contraction operator, since that given any functions $q_1^i$ and $q_2^i$, we have
		\begin{equation}
		\begin{aligned}
		&||H^iq_1^i-H^iq_2^i||_\infty\\
		=&\max\limits_{(s,a^i)}\gamma\left|\mathop{\mathbb{E}}\limits_{a^{-i}\sim\pi^{-i}}\left[\mathop{\mathbb{E}}\limits_{\substack{\tilde{s}\sim P\\\tilde{a}^i\sim\pi^i}}\left(q^i_1(\tilde{s},\tilde{a}^i)-q^i_2(\tilde{s},\tilde{a}^i)\right)\right]\right|\\
		\leq&\max\limits_{(s,a^i)}\gamma\mathop{\mathbb{E}}\limits_{a^{-i}\sim\pi^{-i}}\left[\mathop{\mathbb{E}}\limits_{\substack{\tilde{s}\sim P\\\tilde{a}^i\sim\pi^i}}\left|q^i_1(\tilde{s},\tilde{a}^i)-q^i_2(\tilde{s},\tilde{a}^i)\right|\right]\\
		\leq&\max\limits_{(s,a^i)}\gamma\mathop{\mathbb{E}}\limits_{a^{-i}\sim\pi^{-i}}||q_1^i-q_2^i||_\infty\\
		=&\gamma||q^i_1-q^i_2||_\infty.
		\end{aligned}
		\end{equation}
		Furthermore, from the Bellman equation \eqref{equ:be}, it can be easily found that $Q^i_\pi$ is the fixed point of $H^i$, i.e., $H^iQ_\pi^i=Q_\pi^i$. Therefore, substituting it to \eqref{equ:ef}, we have
		\begin{equation}
		\mathbb{E}\left[F_t^i(s,a^i)|\mathcal{F}_t\right]=(H^iQ_t^i)(s,a^i)-(H^iQ_\pi^i)(s,a^i).
		\end{equation}
		The contraction operator $H^i$ guarantees that
		\begin{equation}
		\mathbb{E}\left[F_t^i(s,a^i)|\mathcal{F}_t\right]\leq \gamma||Q_t^i-Q_\pi^i||_\infty=\gamma||\Delta_t^i||_\infty,
		\end{equation}
		which means Condition \ref{condition2}) is satisfied.
		
		Then we check variance of $F_t^i$
		\begin{equation}
		\begin{aligned}
		var\left[F_t^i(s,a^i)|\mathcal{F}_t\right]=var\left[R(s,a)+\gamma Q_t^i(\tilde{s},\tilde{a}^i)\right].
		\end{aligned}
		\end{equation}
		Due to that $R$ is bounded, $Q^i_\pi$ is also bounded. Besides, $Q_t^i$ can be bounded by $||Q^i_\pi||_\infty+||\Delta_t^i||\infty$. Therefore, some $C>0$ must exist so that
		\begin{equation}
		var\left[F_t^i(s,a^i)|\mathcal{F}_t\right]\leq C(1+||\Delta_t^i||^2_\infty),
		\end{equation}
		which means that Condition \ref{condition3}) is satisfied and also concludes the proof.
	\end{proof}
	
	Theorem \ref{the:pe} shows that by recursively applying \eqref{equ:qupdate}, each agent's estimation $Q^i_t$ asymptotically converge to $Q^i_\pi$. The critics $Q^i$ will be used by the policy improvement step in the next section.
	
	\subsection{Policy Improvement}
	\label{sec:spi}
	In this section, we will derive the policy improvement algorithm by assuming that $Q^i_t$ has converged, i.e., $Q^i_\pi$ will be used to design the algorithm.
	
	In order to stabilize the training process, only one of the agents updates its policy in each iteration. Thus the agents will update their policies in turn, i.e., let $\mod$ denote the modular operator, only the agent $i=\mod(k, N)$ updates its policy at the $k$-th policy improvement step
	\begin{equation}
	\label{equ:pi1}
	\pi_{k+1}^i(s,\cdot)=\mathop{argmax}\limits_{\pi^i(s,\cdot)}\mathop{\mathbb{E}}\limits_{a^i\sim\pi^i}Q_{\pi_k}^i(s,a^i),~\forall s\in \mathcal{S},
	\end{equation}
	and others keeps original policy
	\begin{equation}
	\label{equ:pi2}
	\pi^j_{k+1}=\pi^j_k,~\forall j\neq i.
	\end{equation}
	
	By applying the process of policy improvement, the global policy will be reinforced, which is guaranteed by the theorem below.
	
	\begin{theorem}
		\label{the:pi}
		Let $\pi_k$ denote the global policy at step $k$ and $Q_{\pi_k}^i$ be the action-value function of agent $i$.
		By applying the policy improvement step defined in \eqref{equ:pi1} and \eqref{equ:pi2}, the policy will be reinforced to get higher expected future reward, i.e.,
		\begin{equation}
		V^{\pi_{k+1}}(s)\geq V^{\pi_{k}}(s),~\forall s\in \mathcal{S}.
		\end{equation}
		and the policy $\pi$ will converge.
	\end{theorem}
	
	\begin{proof}
		By applying \eqref{equ:pi1}, for any $s_0\in\mathcal{S}$, we have
		\begin{equation}
		\begin{aligned}
		V^{\pi_k}(s_0)=\mathop{\mathbb{E}}\limits_{a^i_0\sim\pi_k^i}Q_{\pi_k}^i(s_0,a^i_0)\leq \mathop{\mathbb{E}}\limits_{a^i_0\sim\pi_{k+1}^i}Q_{\pi_k}^i(s_0,a^i_0).
		\end{aligned}
		\end{equation}
		Substituting \eqref{equ:be} into $Q_{\pi_k}^i$, we have
		\begin{equation}
		\begin{aligned}
		&\mathop{\mathbb{E}}\limits_{a^i_0\sim\pi_{k+1}^i}Q_{\pi_k}^i(s_0,a^i_0) \\
		=&\mathop{\mathbb{E}}\limits_{\substack{a^i_0\sim \pi_{k+1}^i\\a_0^{-i}\sim\pi^{-i}_k}}\left[R(s_0,a_0)+\gamma\mathop{\mathbb{E}}\limits_{s_1\sim P}V^{\pi_k}(s_1)\right]\\
		=&\mathop{\mathbb{E}}\limits_{\substack{a^i_0\sim \pi_{k+1}^i\\a_0^{-i}\sim\pi^{-i}_{k+1}}}\left[R(s_0,a_0)+\gamma\mathop{\mathbb{E}}\limits_{s_1\sim P}V^{\pi_k}(s_1)\right]\\
		=&\mathop{\mathbb{E}}\limits_{a_0\sim \pi_{k+1}}\left[R(s_0,a_0)+\gamma\mathop{\mathbb{E}}\limits_{s_1\sim P}V^{\pi_k}(s_1)\right].
		\end{aligned}
		\end{equation}
		Thus the relation between $s_0$ and its succeeding state $s_1$ is established
		\begin{equation}
		\label{equ:s0}
		V^{\pi_k}(s_0)\leq \mathop{\mathbb{E}}\limits_{a_0\sim \pi_{k+1}}\left[R(s_0,a_0)+\gamma\mathop{\mathbb{E}}\limits_{s_1\sim P}V^{\pi_k}(s_1)\right].
		\end{equation}
		Similarly, replacing $s_0$ with $s_1$ and replacing $s_1$ with $s_2$, we have
		\begin{equation}
		\label{equ:s1}
		V^{\pi_k}(s_1)\leq \mathop{\mathbb{E}}\limits_{a_1\sim \pi_{k+1}}\left[R(s_1,a_1)+\gamma\mathop{\mathbb{E}}\limits_{s_2\sim P}V^{\pi_k}(s_2)\right].
		\end{equation}
		Substituting it to \eqref{equ:s0}, it can be seen that
		\begin{equation}
		\begin{aligned}
		&V^{\pi_k}(s_0)\\
		\leq &\mathop{\mathbb{E}}\limits_{(a_0,a_1)\sim \pi_{k+1}}\left[R(s_0,a_0)+\gamma R(s_1,a_1)+\gamma^2\mathop{\mathbb{E}}\limits_{s_2\sim P}V^{\pi_k}(s_2)\right].
		\end{aligned}
		\end{equation}
		Repeating this process by expanding $V^{\pi_k}(s_2)$, $V^{\pi_k}(s_3)$ and so on, the following relation can be obtained
		\begin{equation}
		V^{\pi_k}(s_0)\leq \mathop{\mathbb{E}}\limits_{a_t\sim \pi_{k+1}} \sum\limits_{t=0}^\infty\gamma^t R(s_t,a_t) = V^{\pi_{k+1}}(s_0).
		\end{equation}
		Therefore, for any state $s\in \mathcal{S}$, $V^{\pi_{k+1}}(s)\geq V^{\pi_{k}}(s)$, which means the policy is ''improved".
		
		From above results, $\left\{V^{\pi_k}(s)\right\}$ is a non-decreasing sequence, and obviously it has upper bound, since $R(s,a)$ is bounded. Thus the sequence $\left\{V^{\pi_k}(s)\right\}$ will converge, which means the policy $\pi$ will converge, and this also concludes the proof.
	\end{proof}
	
	Theorem \ref{the:pi} guarantees that $\pi$ is monotonically improved, but unfortunately, as most of the policy gradient algorithms, e.g., policy gradient \cite{peters2006policy} and DDPG \cite{silver2014deterministic}, it cannot be guaranteed to converge to the optimal policy. In our experiments, we will show that the sub-optimal policy also performs competitively.
	
	Combining the policy evaluation and policy improvement derived above, we propose the algorithm for real implementation in Algorithm \ref{alg:fsMDP}. We replace $Q_\pi^i$ with its local estimation $Q_t^i$ in the policy improvement step, and typically, the learning rate of the policy evaluation step $\alpha_t^i$ is greatly larger than that of the policy improvement step $\beta_t^i$, which introduce a two time-scales optimization. The actor is optimized along the slower time scale, which guarantees that the global policy is stable in the process of policy evaluation. Furthermore, the critic is optimized along the faster time scale, which leads to more sufficient learning and decreases the estimation error. Besides, all of the agents update their local policies simultaneously, rather than sequentially, to accelerate the learning process. Since that the policies are updated along the slower time scale, although they change policy simultaneously, the environment for each agent is also stable enough for policy evaluation.
	
	\begin{algorithm}
		
		\caption{MARL with discrete state-action-space MDP}
		\begin{algorithmic}
			\label{alg:fsMDP}
			\REQUIRE initial policy $\pi^i_0$ and critic $Q^i_0$, $1\leq i\leq N$; initial state $s_0$; iteration counter $t\gets 0$; step size of critics $\alpha^i_t$ and step size of actor $\beta_t^i$; the threshold for $\zeta$, $T_m$ and $T_M$, which qualify $T_m<T_M$.
			\REPEAT
			\STATE Each agent samples action $a_t^i\sim \pi^i_t(s_t,\cdot)$ to generate $a_t$.
			\STATE Execute $a_t$ and observe $s_{t+1}$, $r_t$.
			\STATE Each agent $i$ executes one step of policy evaluation by following \eqref{equ:qupdate}.
			\STATE Each agent $i$ executes one step of policy improvement
			\begin{equation}
			\label{equ:poi1}
			v^i_t(s_t)=\sum\limits_{a^i\in\mathcal{A}^i}\pi_t^i(s_t,a^i)Q^i_t(s_t,a^i)
			\end{equation}
			\begin{equation}
			\label{equ:poi2}
			\zeta^i(s_t,\cdot)\gets\zeta^i(s_t,\cdot)+\beta_t^i\left[\pi^i_t(s_t,\cdot)\odot\left(Q_t^i(s_t,\cdot)-v^i_t(s_t)\right)\right]
			\end{equation}
			\begin{equation}
			\label{equ:poi3}
			\zeta^i(s_t,\cdot)\gets clip\left\{\zeta^i(s_t,\cdot), T_m, T_M\right\}
			\end{equation}
			\STATE Set iteration counter $t\gets t+1$.
			\UNTIL{Convergence}
		\end{algorithmic}
	\end{algorithm}
	
	Finally, let us explain the policy improvement step \eqref{equ:poi1}, \eqref{equ:poi2} and \eqref{equ:poi3}. Different from \eqref{equ:pi1}, where we directly update the policy to the optimal one under the current situation, here we slowly update the policy $\pi^i$ by gradient ascent. Let
	\begin{equation}
	J(\zeta^i(s_t,\cdot))=\mathop{\mathbb{E}}\limits_{a^i\sim\pi^i_t}Q_t^i(s_t,a^i),
	\end{equation}
	where the relation between $\zeta^i$ and $\pi^i_t$ is defined in \eqref{equ:zeta}. Then the gradient can be calculated by
	\begin{equation}
	\label{equ:graj}
	\begin{aligned}
	\nabla_{\zeta^i}J(\zeta^i(s_t,\cdot))
	=&\nabla_{\zeta^i}\sum\limits_{a^i\in\mathcal{A}^i}\pi_t^i(s_t,a^i)Q_t^i(s_t,a^i)\\
	=&\sum\limits_{a^i\in\mathcal{A}^i}\nabla_{\zeta^i}\pi_t^i(s_t,a^i)Q_t^i(s_t,a^i)\\
	=&\sum\limits_{a^i\in\mathcal{A}^i}\pi_t^i(s_t,a^i)Q_t^i(s_t,a^i)\nabla_{\zeta^i}\log \pi_t^i(s_t,a^i) ,
	\end{aligned}
	\end{equation}
	where the second equation is reached by replacing $f(x)$ with $\pi_t^i(s_t,a^i)$ in $\nabla log(f(x))=\frac{\nabla f(x)}{f(x)}$. The item $\pi_t^i(s_t,a^i)$ is calculated by the soft max function \eqref{equ:zeta}. Thus the gradient can be calculated by
	\begin{equation}
	\label{equ:glog}
	\begin{aligned}
	&\nabla_{\zeta^i}\log \pi_t^i(s_t,a^i)\\
	=&\nabla_{\zeta^i}\zeta^i(s_t,a^i) - \nabla_{\zeta^i}\log\left\{\sum\limits_{\tilde{a}^i\in\mathcal{A}^i}\exp\left\{\zeta^i(s_t,\tilde{a}^i)\right\}\right\}\\
	=&e_{a^i}-\pi^i_t(s_t,\cdot),
	\end{aligned}
	\end{equation}
	where $e_{a^i}$ is an vector in $\mathbb{R}^{|\mathcal{A}^i|}$, whose item in the dimension corresponding to $a^i$ is 1 and others are 0. Then, substituting \eqref{equ:glog} into \eqref{equ:graj} and defining $v_t^i(s_t)$ as \eqref{equ:poi1}, the gradient can be derived by
	\begin{equation}
	\nabla_{\zeta^i}J(\zeta^i(s_t,\cdot))=\pi^i_t(s_t,\cdot)\odot\left(Q_t^i(s_t,\cdot)-v^i_t(s_t)\right).
	\end{equation}
	Therefore, \eqref{equ:poi2} is one step of updating the policy $\pi^i$ with gradient ascent. Finally, \eqref{equ:poi3} is a clip-function that guarantees that each action could be sampled with a positive probability.
	
	In this section, the algorithm for dealing with discrete state-action-space MDP is proposed, and the convergence analysis is given. Here we use tables to store both the Q-function and policy, and the algorithms for policy evaluation and policy improvement are designed correspondingly. However, in more general cases, where either the state space or action space is continuous, the critic $Q^i_\pi(s,a^i):\mathcal{S}\times\mathcal{A}^i\to\mathbb{R}$ and policy $\pi^i(s,a^i):\mathcal{S}\times\mathcal{A}^i\to\mathbb{R}$ cannot be stored by tables. Thus we use function approximation, e.g. NNs, to deal with continuous state-action-space scenarios. The policy evaluation and policy improvement algorithms will also be extended.
	
	\section{MARL with continuous state-action-space MDP}
	\label{sec:csMDP}
	Similar to Section \ref{sec:fsMDP}, let us discuss the policy evaluation and policy improvement algorithms respectively in continuous state-action-space MDP. Besides, in order to improve the data efficiency, we use replay buffers to store historical experience, which will be reused in later updating. Different from Algorithm \ref{alg:fsMDP}, which only uses the samples of current step to update, in order to reuse items in the replay buffers, an additional consensus mechanism is designed to balance the difference between historical and current policies by exchanging information among the communication network $\mathcal{G}$.
	
	When dealing with continuous state-action-space MDP, tables are invalid to represent policies and critics, we instead use NNs to store them. Thus in this section, the notations are slightly different from the previous parts. We use $\pi_\theta$ to denote the global policy parameterized by $\theta$. Similarly, local policy is denoted by $\pi_\theta^i$ and the critic of agent $i$ is parameterized by $\phi^i$, i.e., $Q_{\pi_\theta}^{\phi^i}$.
	
	In this section, let us start by introducing the replay buffers and the consensus mechanism. After that, the policy evaluation and policy improvement mechanism will be introduced.
	
	\subsection{Decentralized Replay Buffer}
	\label{sec:rb}
	We start by introducing the decentralized replay buffer. Similar to most of the centralized off-policy reinforcement learning algorithms, e.g. DQN \cite{mnih2013playing}, historical experiences of agent $i$ at time step $t$ is recorded as $e_t=(s_t,a^i_t,r_t,s_{t+1})$, and each agent is assigned with a replay buffer $D^i=\left\{e_t,e_{t-1},...,e_{t-M+1}\right\}$ to store its own experiences, where $M$ is the memory size.
	
	In the later analysis of Section \ref{sec:pi}, the historical experiences in replay buffer need to be reweighed to evaluate the current policy, and the coefficients for reweighing is related to the difference between current policy $\pi_{\theta}$ and historical policy $\pi_{\hat{\theta}}$. Here we mainly introduce how to estimate the coefficients. Specifically, for each historical experience $e_t$, log-form weight for importance sampling $c_{e_t}^i=\log \frac{\pi_\theta^{-i}(s_t,a_t^{-i})}{\pi_{\hat{\theta}}^{-i}(s_t,a_t^{-i})}$ will be estimated by information exchanging through the communication network $\mathcal{G}$. Let
	\begin{equation}
	\beta_{e_t}^i=\log \frac{\pi_{\theta}^i(s_t, a_t^i)}{\pi_{\hat{\theta}}^i(s_t,a_t^i)},
	\end{equation}
	which can be calculated locally, and the coefficient $c_{e_t}^i$ can be rewritten as
	\begin{equation}
	\label{equ:cet}
	\begin{aligned}
	c_{e_t}^i&=\log \frac{\pi_\theta^{-i}(s_t,a_t^{-i})}{\pi_{\hat{\theta}}^{-i}(s_t,a_t^{-i})}\\
	&=\sum\limits_{j=1}^N\log \frac{\pi_\theta^{j}(s_t,a_t^{j})}{\pi_{\hat{\theta}}^{j}(s_t,a_t^{j})}-\log\frac{\pi_\theta^{i}(s_t,a_t^{i})}{\pi_{\hat{\theta}}^{i}(s_t,a_t^{i})}\\
	&=\sum\limits_{j=1}^N \beta_{e_t}^j - \beta_{e_t}^i,
	\end{aligned}
	\end{equation}
	where $\beta_{e_t}^i$ can be calculated locally and $\sum_{j=1}^N \beta_{e_t}^j$ can be estimated by Lemma \ref{lem:avc}.
	
	\begin{lemma}[\cite{xiao2007distributed}]
		\label{lem:avc}
		For a $N$-agents system connected by the communication network $\mathcal{G}=\left\{\mathcal{X},\mathcal{E}\right\}$, define $\mathcal{N}_i = \left\{j|(i,j)\in \mathcal{E}\right\}$ and $|\mathcal{N}_i|$ as the size of $\mathcal{N}_i$, i.e., the degree of node $i$. Construct a matrix kernel $W\in\mathbb{R}^{N}$
		\begin{equation}
		W_{ij}=\left\{
		\begin{aligned}
		&\frac{1}{d+1}& &i\neq j~and~(i,j)\in\mathcal{E}\\
		&1-\frac{|\mathcal{N}_i|}{d+1}& &i=j\\
		&0& & i\neq j ~ and ~ (i,j)\notin \mathcal{E}
		\end{aligned}
		\right.,
		\end{equation}
		where $d=\max\limits_{i}|\mathcal{N}_i|$ is the degree of $\mathcal{G}$. Then for any initial value $x_0=(x_0^1,x_0^2,...,x_0^N)\in \mathbb{R}^N$, by recursively applying
		\begin{equation}
		\label{equ:xu}
		x_{t+1}=Wx_t,
		\end{equation}
		each agent's local estimation $x^i_t$ will converge to the average of $x_0$, i.e.,
		\begin{equation}
		\lim\limits_{t\to \infty}x_t^i=\frac{1}{N}\sum\limits_{j=1}^Nx_0^j.
		\end{equation}
	\end{lemma}
	
	The proof of Lemma \ref{lem:avc} can be found in \cite{xiao2007distributed} and \cite{scherber2004locally}.
	
	If $x_0^i$ in Lemma \ref{lem:avc} is initialized as $\beta_{e_t}^i$ and update by \eqref{equ:xu} at each time step, average of $\left\{\beta_{e_t}^j, j=1,2,\cdots,N\right\}$ can be evaluate by each agent. Therefore, the quantity $\sum_{j=1}^N \beta_{e_t}^j$ and $c_{e_t}^i$ in \eqref{equ:cet} can be calculated locally.
	
	Because that the policies update in each step, $\beta_{e_t}^i$ also changes. Thus for each experience $e_t$, the $\beta_{e_t}^i$ and $x_{e_t}^i$ update simultaneously, where $x_{e_t}^i$ is the $x_t^i$ assigned for $e_t$. We design the following two steps to update them.
	\begin{enumerate}
		\item Firstly, update $\beta_{e_t}^i$ and $x_{e_t}^i$ with local policy
		\begin{equation}
		\label{equ:bu1}
		\beta_{e_t}^i\gets \log \frac{\pi_{\theta}^i(s_t, a_t^i)}{\pi_{\hat{\theta}}^i(s_t,a_t^i)},
		\end{equation}
		\begin{equation}
		\label{equ:bu2}
		x_{e_t}^i \gets x_{e_t}^i + \left(\beta_{e_t}^i - \hat{\beta}_{e_t}^i\right),
		\end{equation}
		where $\hat{\beta}_{e_t}^i$ is the old value of $\beta_{e_t}^i$ in last step.
		\item Then, execute one step of consensus
		\begin{equation}
		\label{equ:bu3}
		x^i_{e_t}\gets \sum\limits_{j=1}^NW_{ij}x^j_{e_t},~\forall i.
		\end{equation}
	\end{enumerate}
	
	In the first step, both $\beta_{e_t}^i$ and $x_{e_t}^i$ are updated with local information, which means that it can be executed in a distributed manner. Further, in the second step, each agent update $x^i_{e_t}$ with only information from neighbors and itself, since for all the agents $j$ who do not communicate with agent $i$ directly, $W_{ij}=0$. Therefore, the replay buffers can be updated in a distributed manner without violating the rule of privacy.
	
	Notice that in Section \ref{sec:fsMDP}, there is no need to communicate among the agents to reach a consensus, since both the actor and critic are updated with the latest experience generated by the current policies. However, in the continuous state-action-space case, in order to reuse the experience in the replay buffers, the agents should communicate among the graph $\mathcal{G}$ to propagate their local message, in which way to get the weight for each record.
	
	Based on the replay buffer introduced here, we design the policy evaluation and policy improvement algorithms in the next sections.
	
	\subsection{Policy Evaluation}
	\label{sec:pi}
	In this section, we will derive the policy evaluation process under the decentralized setting.
	
	In Section \ref{sec:tpe}, we have derived the policy evaluation algorithms for discrete state-action-space MDP and tabular setting. Here we mainly extend \eqref{equ:qupdate} to fit for NNs-approximation scenarios.
	
	Let us rewrite \eqref{equ:qupdate} as
	\begin{equation}
	\label{equ:qupdate2}
	\begin{aligned}
	Q^i_{t+1}(s_t,a^i_t)=Q_t^i(s_t,a_t^i)+\alpha_t^i(s_t,a_t^i)\left[y_t^i-Q_t^i(s_t,a_t^i)\right],
	\end{aligned}
	\end{equation}
	where $y_t^i = r_t+\gamma Q_t^i(s_{t+1},\tilde{a}^i)$.
	
	Then \eqref{equ:qupdate2} can be seen as one step of updating to minimize the error between $y_t^i$ and $Q^i_t(s_t,a_t^i)$.
	
	When using NNs, this step will be realized by minimizing the loss function
	\begin{equation}
	\label{equ:L}
	L^i(\phi^i) = \frac{1}{2}\mathop{\mathbb{E}}\limits_{(s,a^i)\sim D^i}\left[(Q_{\pi_\theta}^{\phi^i}(s,a^i) - y_t^i)^2\right],
	\end{equation}
	where by slightly abusing of notation, $y_t^i$ in following parts denotes the target value calculated by
	\begin{equation}
	y_t^i=\mathop{\mathbb{E}}\limits_{a^{-i}\sim \pi_\theta^{-i}}\left[R(s,a)+\gamma \mathop{\mathbb{E}}\limits_{\substack{
			\tilde{s}\sim P\\
			\tilde{a}^i\sim\pi^i_\theta
	}}Q_{\pi_\theta}^{\hat{\phi}^i}(\tilde{s},\tilde{a}^i)\right],
	\end{equation}
	where $Q_{\pi_\theta}^{\hat{\phi}^i}$ is the target NNs for stabilizing the training process. Similar technique also can be found in DQN \cite{mnih2015human}, DDPG \cite{silver2014deterministic}. In fact, \eqref{equ:L} is constructed by the residual error of \eqref{equ:be}.
	
	In each step, we update $\phi^i$ by one step of gradient descent:
	\begin{equation}
	\begin{aligned}
	\nabla_{\phi^i}L^i(\phi^i)=\mathop{\mathbb{E}}\limits_{(s,a^i)\sim D^i}\left[\left(Q_{\pi_\theta}^{\phi^i}(s,a^i) - y_t^i\right)\nabla_{\phi^i}Q_{\pi_\theta}^{\phi^i}(s,a^i)\right],
	\end{aligned}
	\end{equation}
	which can be rewritten as
	\begin{equation}
	\label{equ:gL}
	\begin{aligned}
	\nabla_{\phi^i}L^i(\phi^i)=\mathop{\mathbb{E}}\limits_{\substack{(s,a^i)\sim D^i\\a^{-i}\sim \pi^{-i}_{\theta}}}\left[\left(Q_{\pi_\theta}^{\phi^i}(s,a^i) -q^i(s,a) \right)\nabla_{\phi^i}Q_{\pi_\theta}^{\phi^i}(s,a^i)\right],
	\end{aligned}
	\end{equation}
	where
	\begin{equation}
	\label{equ:q}
	q^i(s,a)=\left(R(s,a)+\gamma \mathop{\mathbb{E}}\limits_{\substack{
			\tilde{s}\sim P\\
			\tilde{a}^i\sim\pi^i_\theta
	}}Q_{\pi_\theta}^{\hat{\phi}^i}(\tilde{s},\tilde{a}^i)\right).
	\end{equation}
	
	The expectation terms in \eqref{equ:gL} are hard to be calculated precisely and usually estimated by samples taken from corresponding distributions.
	
	We intend to evaluate it with samples in the replay buffer $D^i$. Thus, we adjust the distribution of samples in $D^i$ to fit the distributions in \eqref{equ:gL} and \eqref{equ:q} by importance sampling. First, the two distributions in the expectation of $q^i(s, a)$. $\tilde{s}\sim P$ is naturally satisfied by samples in $D^i$, since that the system dynamics are independent of policies. $\tilde{a}^i\sim\pi^i_\theta$ can be satisfied by taking samples freshly from current local policy $\pi_\theta^i$ without calling for others' policies. Thus the expectation terms in $q^i(s,a)$ can be evaluated by samples from the replay buffer. Furthermore, the first distribution $(s,a^i)\sim D^i$ is also naturally satisfied. However, the second distribution $a^{-i}\sim\pi_\theta^{-i}$ needs to be evaluated by samples taken from current policy $\pi_\theta^{-i}$, but samples in $D^i$ is taken from historical policies, which means that samples in replay buffers cannot be directly used to estimate $\nabla_{\phi^i}L^i(\phi^i)$. We use the technique of importance sampling to solve it.
	
	Here we firstly introduce the importance sampling.
	
	\begin{lemma}[\cite{glynn1989importance}]
		\label{lem:is}
		There are two distributions over a random variable $X$, whose probability density functions are $P(X)$ and $Q(X)$. If $Q(X) > 0$ holds for any value of $X$, then for a function $f(x)$, the following relation holds
		\begin{equation}
		\mathop{\mathbb{E}}\limits_{P(X)}\left[f(X)\right]=\mathop{\mathbb{E}}\limits_{Q(X)}\left[f(X)\frac{P(X)}{Q(X)}\right]\approx \frac{1}{n}\sum\limits_{i=1}^n c(x^i)f(x^i).
		\end{equation}
		where $c(x^i)=\frac{P(x^i)}{Q(x^i)}$ is the coefficient on sample $x^i$.
	\end{lemma}
	
	Lemma \ref{lem:is} can be simple verified by
	\begin{equation}
	\begin{aligned}
	&\mathop{\mathbb{E}}\limits_{P(X)}\left[f(X)\right]
	=\int_xP(x)f(x)dx\\
	=&\int_xQ(x)\frac{P(x)}{Q(x)}f(x)dx
	=\mathop{\mathbb{E}}\limits_{Q(X)}\left[f(X)\frac{P(X)}{Q(X)}\right].
	\end{aligned}
	\end{equation}
	Lemma \ref{lem:is} shows that by assigning reweighing coefficients $c(x)$ to each sample $x$, the expectation over a distribution $P(X)$ can be estimated by samples taken from another distribution $Q(X)$. Thus the expectation term over $a^{-i}\sim\pi_{\theta}^{-i}$ in \eqref{equ:gL} can be estimated by samples taken from old policies $a^{-i}\sim\pi_{\hat{\theta}}^{-i}$, by assigning coefficient $\frac{\pi_{\theta}^{-i}(s_t,a_t^{-i})}{\pi_{\hat{\theta}}^{-i}(s_t,a_t^{-i})}$ to each experience $e_t$, where $\pi_{\hat{\theta}}^{-i}(s_t,a_t^{-i})$ is the probability density of the policy when generating $e_t$. Besides, the coefficient $\pi_{\hat{\theta}}^{-i}(s_t,a_t^{-i})$ can be estimated by $e^{c_{e_t}^i}$, where $c_{e_t}^i$ is defined in Section \ref{sec:rb}.
	
	Therefore, sample a batch $\mathcal{B}^i$ from replay buffer $D^i$ and the gradient of $L^i$ with respect to $\phi^i$ can be estimated by
	\begin{equation}
	\label{equ:g1}
	\begin{aligned}
	&\nabla_{\phi^i}L^i(\phi^i)\\
	\approx&\frac{1}{|\mathcal{B}^i|}\sum\limits_{e_t\in\mathcal{B}^i}e^{c_{e_t}^i}\left(Q_{\pi_\theta}^{\phi^i}(s_t,a_t^i)-\hat{q}^i(s_t,a_t)\right)\nabla_{\phi^i}Q_{\pi_\theta}^{\phi^i}(s_t,a_t^i).
	\end{aligned}
	\end{equation}
	where
	\begin{equation}
	\label{equ:g2}
	\hat{q}^i(s_t,a_t)=R(s_t,a_t)+\gamma Q_{\pi_\theta}^{\hat{\phi}^i}(s_{t+1},\tilde{a}^i)
	\end{equation}
	is an estimator of $q^i(s_t,a_t)$ by taking a sample $\tilde{a}^i$ freshly from current local policy $\pi^i_\theta$.
	
	In each step of updating, agent $i$ locally generates a batch $\mathcal{B}^i$ and estimates the gradient of $L^i$ with respect to parameters of its critic $\phi^i$ with \eqref{equ:g1} and \eqref{equ:g2}. Then, update $\phi^i$ with one step of stochastic gradient descent.
	
	We have introduced the policy evaluation algorithm when using NNs as function approximation, and in the next section, the policy improvement algorithm will be proposed.
	
	\subsection{Policy Improvement}
	In Section \ref{sec:spi}, we proposed the policy improvement algorithm for discrete state-action-space MDP and tabular setting. Here it will be extended to fit for NNs-approximation scenarios.
	
	The policies should optimize to maximize the state-value function under each state $s$, $V^{\pi_\theta}(s)$, as in \eqref{equ:pi1} and \eqref{equ:pi2}. Thus here we set the objective function as
	\begin{equation}
	J(\theta)=\mathop{\mathbb{E}}\limits_{s\sim D^i}V^{\pi_{\theta}}(s),
	\end{equation}
	where we use a batch of state sampled from the replay buffer $D^i$ to reduce variance when estimating gradient of $J$ with respect to $\theta$.
	
	In order to execute the policy improvement process in a distributed manner, each agent calculates $J$ with its local action-value function, thus according to \eqref{equ:vp}, objective function of agent $i$ is
	\begin{equation}
	\label{equ:ji}
	J^i(\theta^i)=\mathop{\mathbb{E}}\limits_{\substack{s\sim D^i\\\tilde{a}^i\sim{\pi_\theta^i}}}Q_{\pi_\theta}^{\phi^i}(s,\tilde{a}^i),
	\end{equation}
	where $\theta^i$ is the parameter of $\pi^i_\theta$.
	
	Then agent $i$ calculate the gradient of $J^i$ with respect to $\theta^i$ and update its local policy with one step of gradient descent. We utilize the reparameterization trick in SAC \cite{haarnoja2018soft} to calculate the gradient. Specifically, we use squashed Gaussian policy
	\begin{equation}
	a^i=f_{\theta^i}(\xi,s)=\tanh\left(\mu_\theta^i(s)+\sigma_\theta^i(s)\odot\xi\right),
	\end{equation}
	where both $\mu_\theta^i$ and $\sigma_\theta^i$ are parts of $\pi_{\theta^i}$ and $\xi\in\mathbb{R}^{|\mathcal{A}^i|}$ is a random vector follows $\xi\sim\mathcal{N}(0,I)$. $\mu_\theta^i(s)\in\mathbb{R}^{|\mathcal{A}^i|}$ and $\sigma_\theta^i(s)\in\mathbb{R}^{|\mathcal{A}^i|}$ map the standard Gaussian distribution to $\mathcal{N}\left(\mu_\theta^i(s),diag\left(\sigma_\theta^i(s)\right)\right)$, which $a^i$ is sampled from.
	
	By the reparameterization trick, \eqref{equ:ji} can be rewritten as
	\begin{equation}
	j^i(\theta^i)=\mathop{\mathbb{E}}\limits_{\substack{s\sim D^i\\ \xi\sim \mathcal{N}}}Q_{\pi_\theta}^{\phi^i}(s,f_{\theta^i}(\xi,s)).
	\end{equation}
	Then we have
	\begin{equation}
	\label{equ:nj}
	\nabla_{\theta^i}J(\theta^i)=\mathop{\mathbb{E}}\limits_{\substack{s\sim D^i\\ \xi\sim \mathcal{N}}}\nabla_{\tilde{a}^i}Q_{\pi_\theta}^{\phi^i}(s,\tilde{a}^i)\nabla_{\theta^i}f_{\theta^i}(\xi,s),
	\end{equation}
	where $\tilde{a}^i=f_{\theta^i}(\xi,s)$. Finally, since \eqref{equ:nj} contains expectation terms over sophisticated distributions, we use samples taken correspondingly to estimate it
	\begin{equation}
	\label{equ:ga}
	\nabla_{\theta^i}J(\theta^i)\approx\frac{1}{|\mathcal{B}^i|}\sum\limits_{e_t\in \mathcal{B}^i}\nabla_{\tilde{a}^i}Q_{\pi_\theta}^{\phi^i}(s_t,\tilde{a}^i)\nabla_{\theta^i}f_{\theta^i}(\xi,s_t),
	\end{equation}
	where $\mathcal{B}^i$ is randomly sampled from $D^i$.
	Finally, run one step of stochastic gradient descent to update the local policy parameters $\theta^i$.
	
	\begin{algorithm}
		\caption{MARL with continuous state-action-space MDP}
		\begin{algorithmic}
			\label{alg:csMDP}
			\REQUIRE initial policy parameters $\theta^i$ and critic parameters $\phi^i$, $1\leq i\leq N$; initial state $s_0$; iteration counter $t\gets 0$; step size of critics $\alpha$ and step size of actor $\beta$; step size of target networks $\epsilon$.
			\REPEAT
			\STATE Each agent samples action $a_t^i\sim \pi^i_\theta(s_t,\cdot)$ to generate $a_t$.
			\STATE Execute $a_t$ and observe $s_{t+1}$, $r_t$.
			\STATE Each agent stores $e_t=(s_t,a_t^i,r_t,s_{t+1})$ and the probability $\pi_\theta^i(s_t,a_t^i)$ to $D^i$. Initialize $\beta_{e_t}^i=0$ and $x_{e_t}^i=0$.
			\STATE Each agent samples $\mathcal{B}^i$ from $D^i$ randomly.
			\STATE \textbf{Policy evaluation step}. Each agent updates $Q_{\pi_\theta}^{\phi^i}$ with one step of gradient descent by following \eqref{equ:g1} and \eqref{equ:g2}.
			\STATE \textbf{Policy improvement step}. Each agent updates policy $\pi_\theta^i$ with one step of gradient ascent by following \eqref{equ:ga}
			\STATE \textbf{Consensus step}. Each agent $i$ updates $\beta_{e_t}^i$ and $x_{e_t}^i$ for all $e_t\in D^i$ locally by following \eqref{equ:bu1} and \eqref{equ:bu2}, and then exchange information with neighbors by following \eqref{equ:bu3}.
			\STATE Each agent update its target network $\hat{\phi}^i\gets(1-\epsilon)\hat{\phi}^i+\epsilon \phi^i$.
			\STATE Set iteration counter $t\gets t+1$.
			\UNTIL{Convergence}
		\end{algorithmic}
	\end{algorithm}
	
	The pseudo-code for dealing with continuous state-action-space MDP is shown in Algorithm \ref{alg:csMDP}. The policy evaluation step and policy improvement step are alternatively executed until convergence.
	
	In this section, we proposed reinforcement learning algorithms for discrete and continuous state-action-space multi-agent MDP. Convergence analysis is given under tabular setting. In the next section, we will design some experiments to show their effectiveness.
	
	\section{Simulation Results}
	\label{sec:sr}
	In this section, we design experiments to show the effectiveness of the proposed algorithms.
	
	Firstly, we apply Algorithm \ref{alg:fsMDP} to randomly generated discrete state-action-space MDP, whose state space $\mathcal{S}=\left\{1, 2, \dots, M\right\}$, where $M = 100$ in our experiments. The state transition matrix $P$ is generated by random numbers $P(s,a,\tilde{s})\sim U(0, 1)$ and then normalized. The reward function $R(s,a)$ is sampled from normal distribution $N(0, 1)$. Besides, in our experiments, each agent has 3 actions, i.e., $|\mathcal{A}^1|=|\mathcal{A}^2|=\dots=|\mathcal{A}^N|=3$. Finally, we compare Algorithm \ref{alg:fsMDP} with a centralized actor-critic algorithm and Q-learning algorithm under $N=4$ and $N=5$ agents scenarios respectively, where communication network $\mathcal{G}$ is a ring. The centralized algorithm (termed as centralized AC) is a special case of Algorithm \ref{alg:fsMDP} with only 1 global agent, whose action is the combination of all the agents' actions. The Q-learning algorithm also runs by contaminating all the agents as a global agent.
	
	We found that the proposed distributed reinforcement learning algorithm, Algorithm \ref{alg:fsMDP}, converges faster compared with centralized AC and Q-learning. Algorithm \ref{alg:fsMDP} is more suitable for larger state-action space. In these two scenarios, the size of the state-action space is $100\times3^4=8,100$ and $100\times3^5=24,300$, respectively. In the scenario with smaller state-action space (Fig. \ref{fig:fs_10_3_4}), Algorithm \ref{alg:fsMDP} converges faster but falls into the local minimum, thus the final return (''return'' is defined as the accumulated reward in the whole episode) is worse than that of the centralized AC. However, it also achieves better performance than Q-learning. In the scenario with larger state-action space (Fig. \ref{fig:fs_10_3_5}), Algorithm \ref{alg:fsMDP} achieves the best performance from the aspects of both convergence rate and final return. Because the searching space is decomposed and each agent only focuses on its own part, the convergence rate is improved.
	
	\begin{figure*}
		\centering
		\subfigure[$N=4$]{
			\includegraphics[scale=0.35]{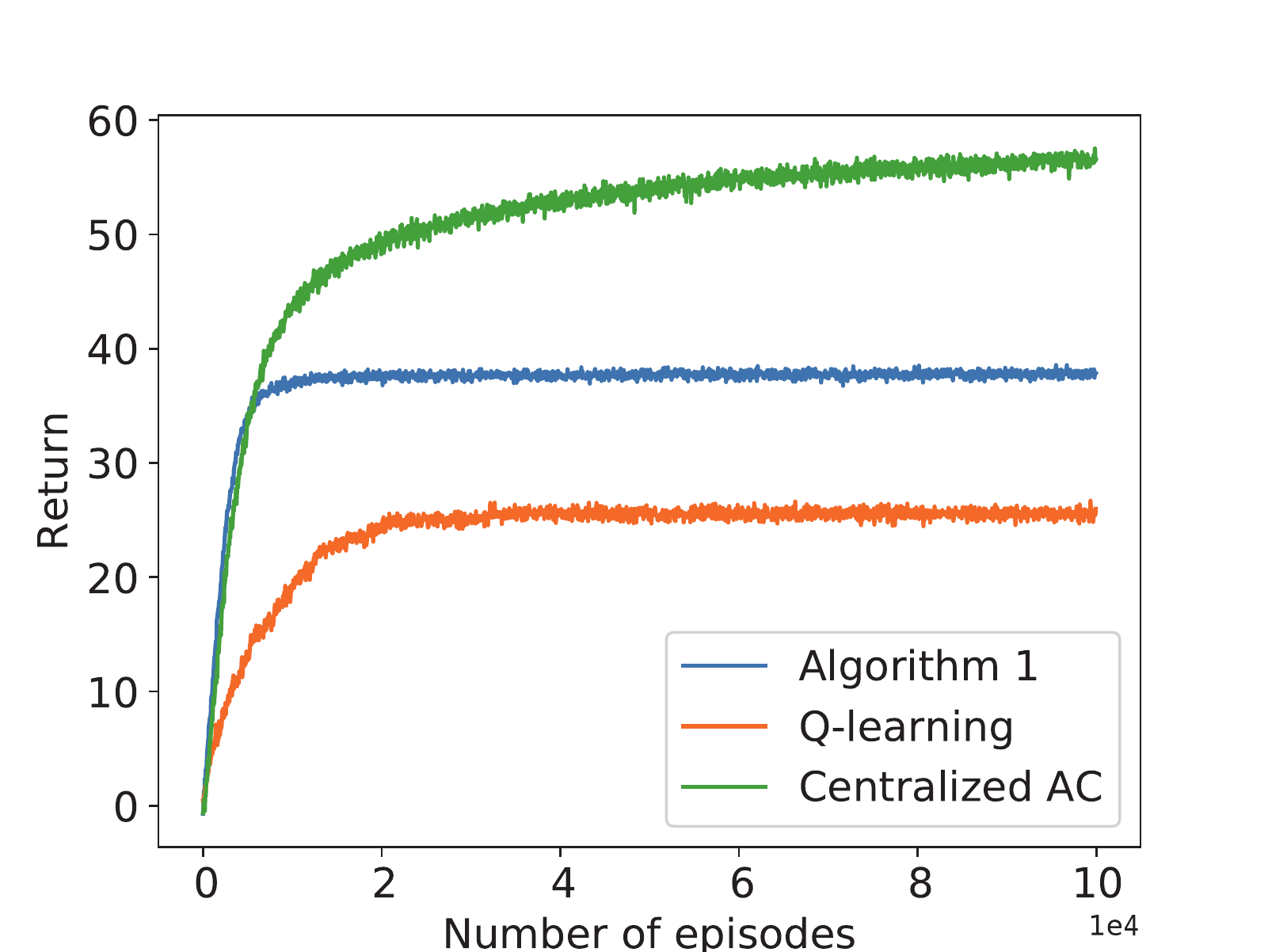}\label{fig:fs_10_3_4}}
		\hspace{1pt}
		\subfigure[$N=5$]{
			\includegraphics[scale=0.35]{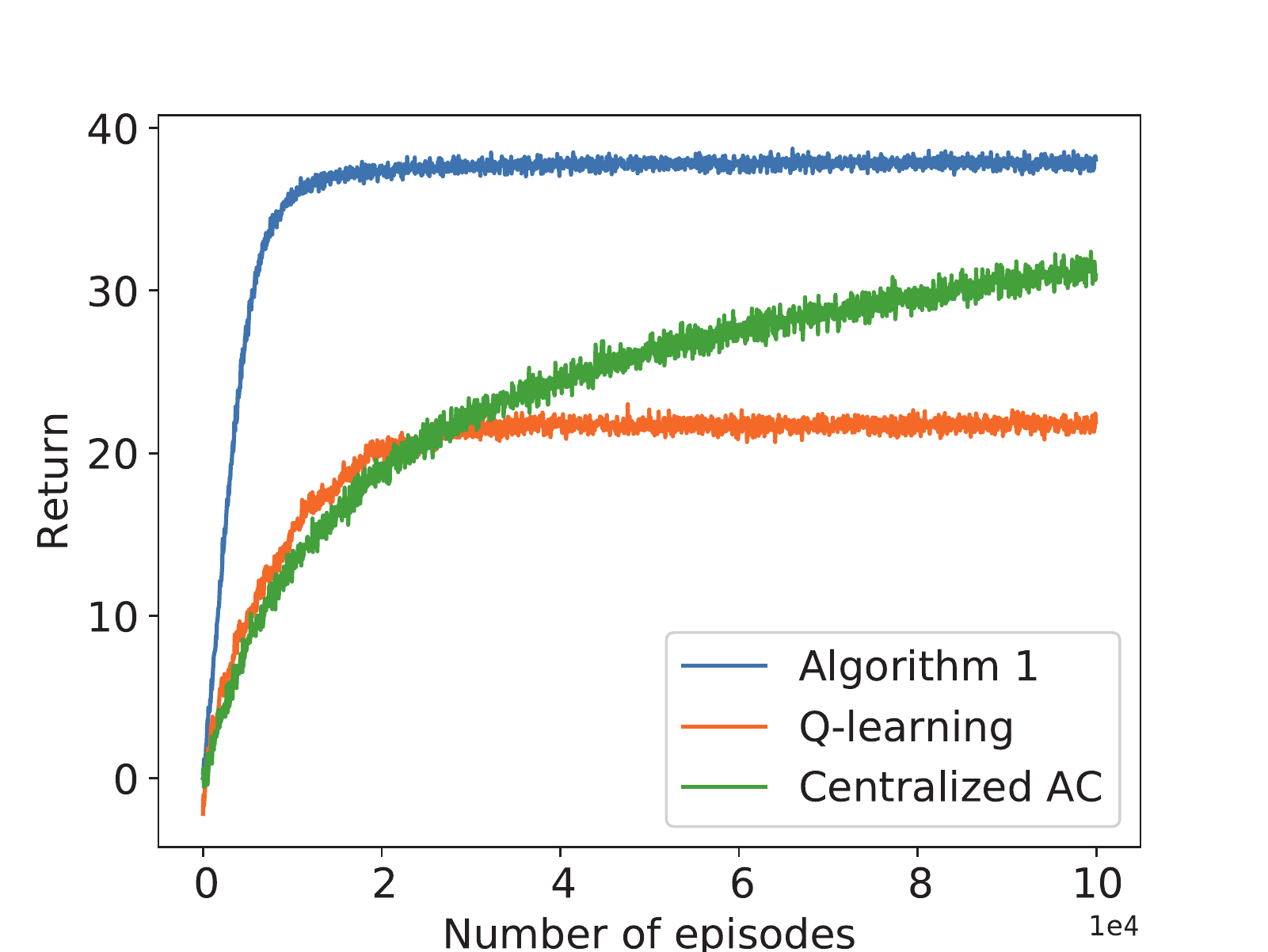}\label{fig:fs_10_3_5}}
		\hspace{1pt}
		\subfigure[simple spread]{
			\includegraphics[scale=0.35]{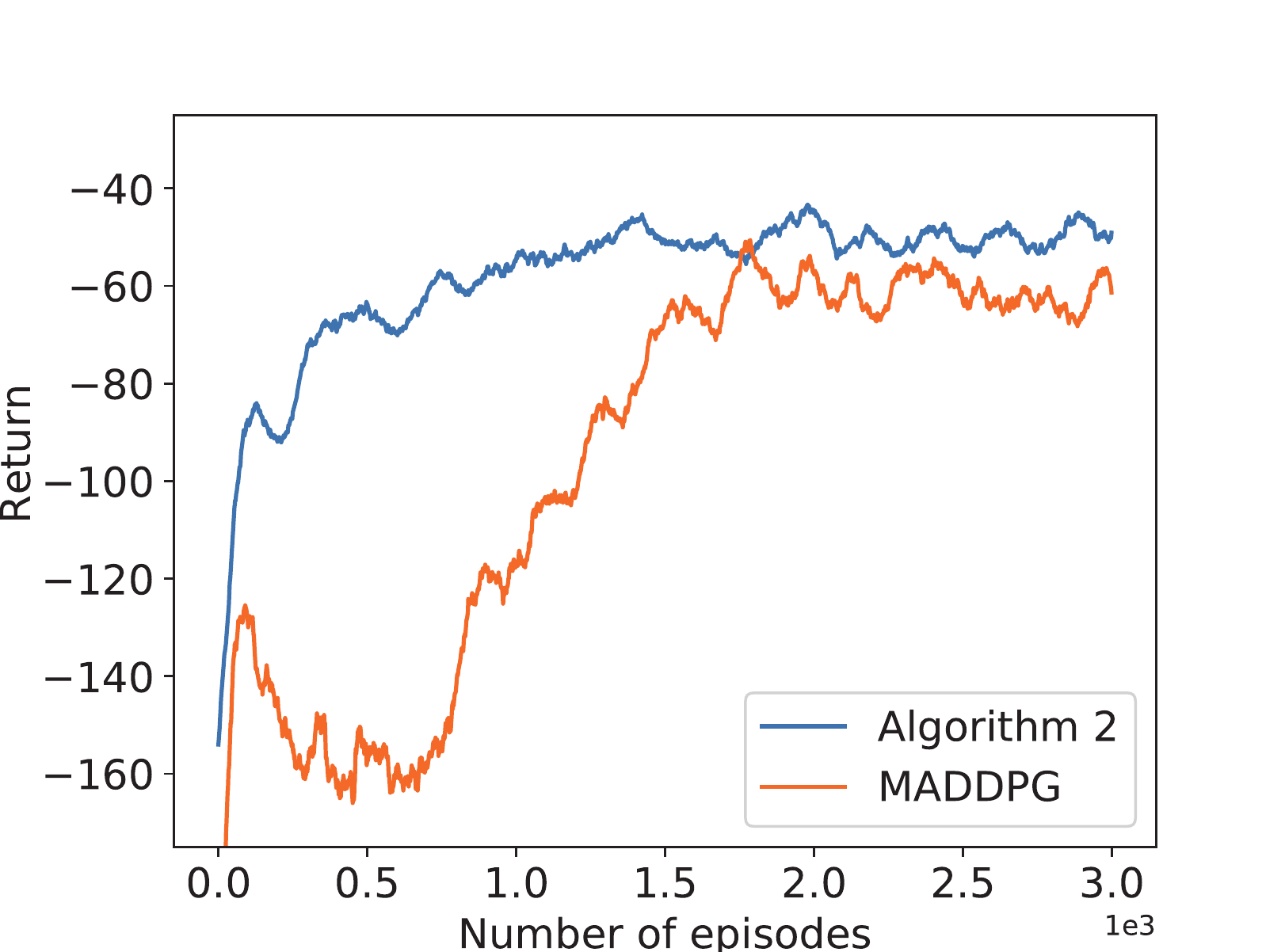}\label{fig:cs_MDP}}
		\caption{(a) and (b) are learning curves for $N=4$ and $N=5$ environments respectively, which are generated by averaging 5 times of independent random experiments. (c) is the learning curves for simple spread environment.}
		\label{fig:fs}
	\end{figure*}
	
	Above we show the results from discrete state-action-space MDP. A more challenging case is MDP with continuous state-action space since the searching space is continuous and has infinite state-action pairs. We further design experiments to Algorithm \ref{alg:csMDP}. Then ''simple spread'' environment is modified from \cite{lowe2017multi}. As illustrated in Fig. \ref{fig:cs_MDP}, there are $N$ agents aiming to reach their targets without collision. The communication network is a ring. The instantaneous reward of each agent is composed of 2 parts, the negative value of the distance to its target and -1 if it collides with other agents. The global reward is calculated by averaging the rewards of all the agents. The global state contains the location of all the agents and targets. The action of each agent is a five-dimensional vector with continuous value range. We set $N=3$ in our experiment.
	
	\begin{figure}
		\centering
		\subfigure[]{
			\includegraphics[scale=0.35]{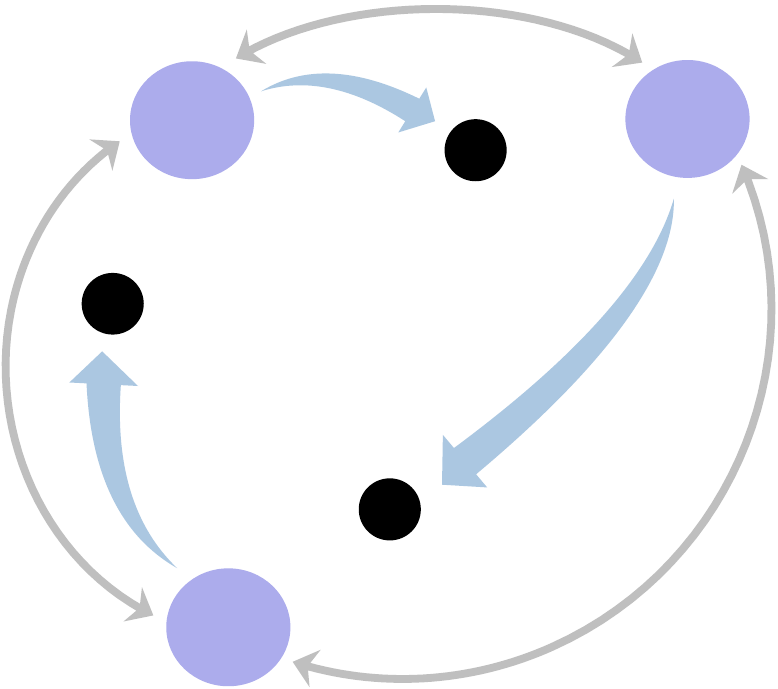}\label{fig:simple_spread}}
		\hspace{20pt}
		\subfigure[]{
			\includegraphics[scale=0.12]{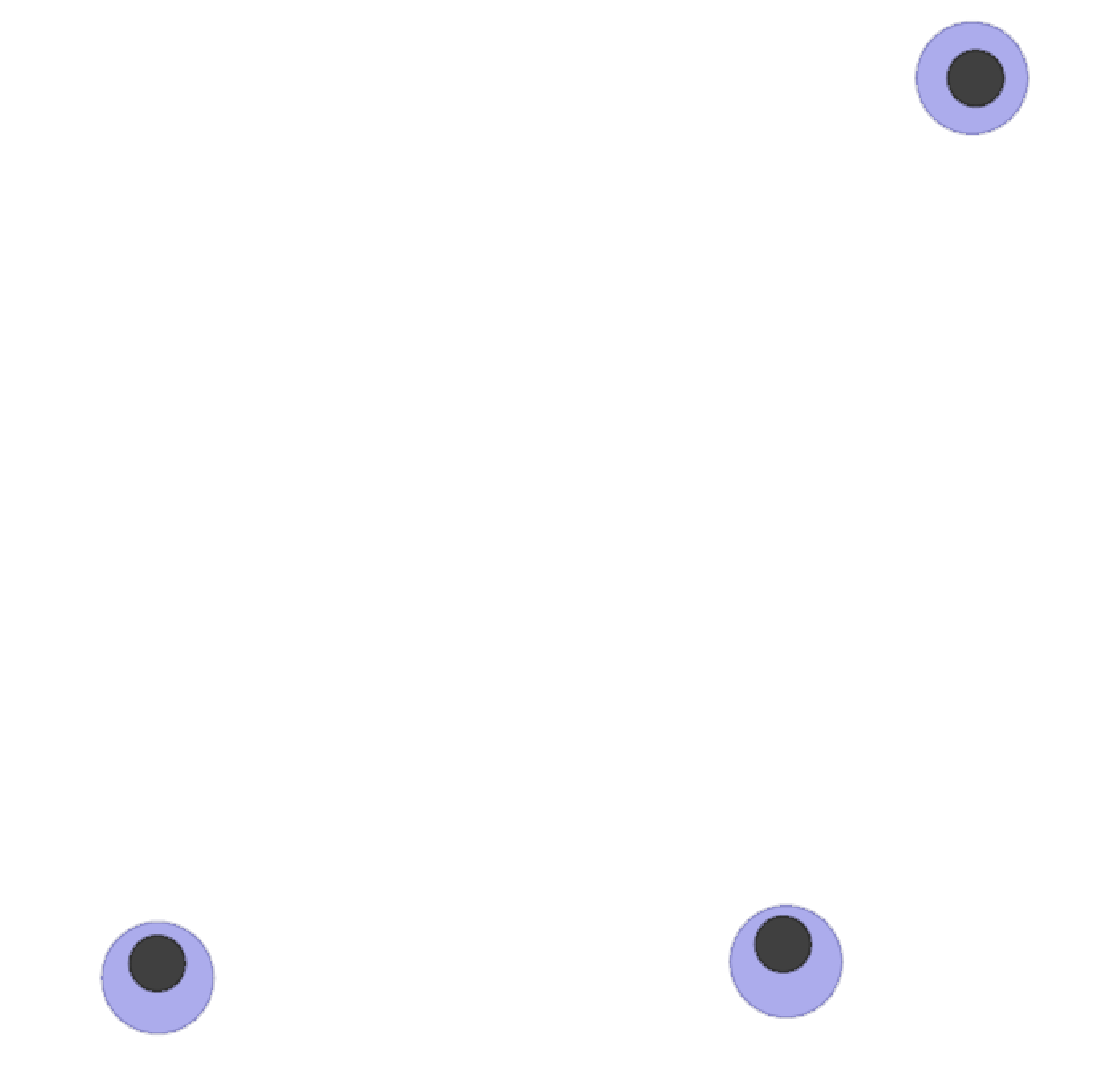}\label{fig:snapshot}}
		\caption{(a) Snapshot of the simple spread environment. The purple balls are the agents and the black dots are their targets. The blue arrows indicate their correspondence. The gray bi-directional arrows are the edges of the communication network $\mathcal{G}$. (b) Snapshot of the final policy.}
	\end{figure}
	
	Algorithm \ref{alg:csMDP} is compared with MADDPG, a centralized MARL algorithm. The learning curves are shown in Fig. \ref{fig:cs_MDP}. It can be found that the proposed Algorithm \ref{alg:csMDP} achieves a higher final return, and the convergence rate also outperforms MADDPG. Besides, MADDPG is an algorithm of centralized training and decentralized execution, which means that it requires the policy of other agents in the training process. However, others' policies may violate privacy in our setting. Algorithm \ref{alg:csMDP} is fully decentralized, which does not require the policy of others during either the learning or the execution process.
	
	Here a snapshot of the final policy is given in Fig. \ref{fig:snapshot}. After training, each agent can find its target precisely without collision.
	
	In this section, simulation results are shown. Besides, we compare them with relative algorithms, e.g., Centralize AC, Q-learning, and MADDPG. The results show that the proposed algorithms perform well from the aspects of both convergence rate and final return, especially under the limitation of privacy.
	
	\section{Conclusion}
	\label{sec:con}
	In this paper, we address the problem of Multi-Agent Reinforcement Learning (MARL). Especially, we focus on a fully decentralized setting, under which each agent has no permission to use others' policies for the purpose of protecting privacy. Each agent makes individual decisions locally, but in the training process, they can exchange information with their neighbors among a communication network. We firstly design the algorithm for discrete state-action-space MDP and give the convergence analysis. Then, we extend it to continuous state-action-space scenarios. Finally, we design experiments to show the effectiveness of the proposed algorithms. By comparing with relative algorithms, all the proposed algorithms are shown to work well from the aspects of both convergence rate and final return. So far, our algorithms mainly work under cooperative MARL problems, and it is an interesting direction in future work to extend them to competitive MARL problems. Besides, how to reach consensus with less information exchanging will also be studied in future work. Since in real scenarios, the capacity of the communication channel is usually limited.
	
	\bibliographystyle{IEEEtran}
	\bibliography{references}{}
	
\end{document}